\documentclass[pmlr]{jmlr}% new name PMLR (Proceedings of Machine Learning Research)

\usepackage{longtable}% for long tables

\usepackage{booktabs}
\usepackage[load-configurations=version-1]{siunitx} % newer version

\theorembodyfont{\upshape}
\theoremheaderfont{\scshape}
\theorempostheader{:}
\theoremsep{\newline}

\usepackage{amssymb}
\usepackage{amsmath}
\usepackage[noend]{algpseudocode}
\usepackage{tikz}
\usetikzlibrary{shapes.geometric, arrows}
\usetikzlibrary {arrows.meta}

\newcommand{\trans}{\delta}
\newcommand{\init}{\iota}
\newcommand{\alphabet}{\Sigma}
\newcommand{\states}{Q}

\newcommand{\wordletter}[2]{#1{[#2]}}

\newcommand{\pathto}[2]{{\xrightarrow[]{{#1}}}}

\newcommand{\emptyword}{\varepsilon}
\newcommand{\finwords}{\alphabet^*}
\newcommand{\poswords}{\alphabet^+}

\newcommand{\inits}{I}
\newcommand{\run}{\rho}
\newcommand{\naturals}{\mathbb{N}}

\newcommand{\size}[1]{|#1|}

\newcommand{\setnocond}[1]{\{#1\}}
\newcommand{\setcond}[2]{\{\, #1 \mid #2 \,\}}

\newcommand{\D}{\mathcal{D}}
\renewcommand{\O}{\mathcal{O}}
\newcommand{\A}{\mathcal{A}}
\newcommand{\M}{\mathcal{M}}
\newcommand{\N}{\mathcal{N}}
\renewcommand{\P}{\mathcal{P}}
\renewcommand{\S}{\mathcal{S}}
\newcommand{\T}{\mathcal{T}}

\newcommand{\posset}{S^{+}}
\newcommand{\negset}{S^{-}}

\newcommand{\odd}[1]{\textsf{odd}(#1)}
\newcommand{\even}[1]{\textsf{even}(#1)}
\newcommand{\opp}{\textsf{opp}}

\newcommand{\canoEq}{\sim}

\newcommand{\var}[1]{\textit{#1}}
\newcommand{\haschild}{\textsf{has\_children}}
\newcommand{\register}{\textsf{replace\_or\_register}}
\newcommand{\lastchild}{\textsf{last\_child}}

\newcommand{\pref}{\mathsf{prefixes}}
\newcommand{\tool}{\textsf{DFAMiner}}
\newcommand{\dfaind}{\textsf{DFA-Inductor}}
\newcommand{\dfaid}{\textsf{DFA-Identify}}

\newtheorem{fact}{Fact}

\jmlrvolume{1}
\jmlryear{2024}
\jmlrworkshop{5th Workshop on Learning and Automata (LearnAut)}

\title[\tool: Mining minimal DFAs]{\tool: Mining minimal separating DFAs from labelled samples}

\author{
\Name{Daniele Dell'Erba}\Email{dde@liverpool.ac.uk} \\
\Name{Yong Li} 
\Email{liyong@liverpool.ac.uk}\\
\Name{Sven Schewe} \Email{svens@liverpool.ac.uk}\\
\addr University of Liverpool, UK}

\begin{document}

\maketitle

\begin{abstract}
    We propose \tool, a passive learning tool for learning minimal separating deterministic finite automata (DFA) from a set of labelled samples.
    Separating automata are an interesting class of automata that occurs generally in regular model checking and has raised interest in foundational questions of parity game solving.
    We first propose a simple and linear-time algorithm that incrementally constructs a three-valued DFA (3DFA) from a set of labelled samples given in the usual lexicographical order.
    This 3DFA has accepting and rejecting states as well as don't-care states, so that it can exactly recognise the labelled examples.
    We then apply our tool to mining a minimal separating DFA for the labelled samples by minimising the constructed automata via a reduction to solving SAT problems.
    Empirical evaluation shows that our tool outperforms current state-of-the-art tools significantly on standard benchmarks for learning minimal separating DFAs from samples.
    Progress in the efficient construction of separating DFAs can also lead to finding the lower bound of parity game solving, where we show that \tool\ can create optimal separating automata for simple languages with up to 7 colours. 
    Future improvements might offer inroads to better data structures.
\end{abstract}

\section{Introduction}
\label{sec:intro}

The task of inferring a minimum-size separating automaton from two disjoint sets of samples has gained much attention from various fields, including computational biology~\citep{DBLP:journals/pr/Higuera05}, inference of network invariants~\citep{DBLP:conf/cade/GrinchteinLP06}, regular model checking~\citep{DBLP:conf/atva/Neider12}, and reinforcement learning~\citep{DBLP:conf/fmcad/LaufferYVSS22}.
More recently, this problem has also arisen in the context of parity game solving~\citep{BC18}, where separating automata can be used to decide the winner.
The breakthrough quasi-polynomial algorithm \citep{CJKLS17}, for example, can be viewed as producing such a separating automaton, and under additional constraints, quasi-polynomial lower bounds can be established, too~\citep{CDFJLP19,CJKLS17}.
These applications can be formalised as seeking the minimum-size of DFAs, known as the Min-DFA inference problem, from positive and negative samples.

The Min-DFA inference problem was first explored in~\citep{DBLP:journals/tc/BiermannF72,DBLP:journals/iandc/Gold78}. 
Due to its high (NP-complete) complexity, researchers initially focused on either finding local optima through state merging techniques~\citep{DBLP:conf/icgi/LangPP98,Oncina-Garcia-92,DBLP:journals/pr/BugalhoO05}, or seeking theoretical aspects such as reduction to graph colouring problems~\citep{DBLP:conf/iwGi/CosteN97}.
Notably, it has been shown that there is no efficient algorithm to find approximate solutions~\citep{DBLP:journals/jacm/PittW93}.

With the increase in computational power and efficiency of Boolean Satisfiability (SAT) solvers, research has shifted towards practical and exact solutions to the Min-DFA inference problem.
Several tools have emerged in the literature, including \textsf{ed-beam}/\textsf{exbar}~\citep{DBLP:conf/icgi/HeuleV10}, \textsf{FlexFringe}~\citep{DBLP:conf/icsm/VerwerH17}, \dfaind~\citep{DBLP:conf/lata/UlyantsevZS15,DBLP:conf/lata/ZakirzyanovMIUM19} and \dfaid~\citep{DBLP:conf/fmcad/LaufferYVSS22}.

The current practical and exact solutions to the Min-DFA inference problem typically involve two steps:
(1) Construct the augmented prefix tree acceptor (APTA~\citep{DBLP:conf/icgi/CosteN98}) that recognises the given samples, and (2) minimise the APTA to a Min-DFA by a reduction to SAT~\citep{DBLP:conf/icgi/HeuleV10}.
Recent enhancements of this approach focus on the second step, including techniques like symmetry breaking~\citep{DBLP:conf/icgi/HeuleV10,DBLP:conf/lata/UlyantsevZS15} and compact SAT encoding~\citep{DBLP:conf/icgi/HeuleV10,DBLP:conf/lata/ZakirzyanovMIUM19}.
Additionally, there is an approach on the incremental SAT solving technique specialised for the Min-DFA inference problem, where heuristics for assigning free variables have also been proposed~\citep{DBLP:conf/sefm/AvellanedaP19}.
However, their implementation relies heavily on \textsf{MiniSAT}~\citep{DBLP:conf/sat/EenS03}.
We believe that, in order to take advantage of future improvements of SAT solvers, it is better to use a SAT solver as a black-box tool.
We note that the second step has also been encoded as a Satisfiability Modulo Theories problem~\citep{DBLP:conf/lata/SmetsersFV18}, which can also be improved by our contribution to the first step.

The second step is typically the bottleneck in the workflow.
It is known that the number of boolean variables used in the SAT problem is polynomial in the number of states of the APTA.
Smaller APTAs naturally lead to easier SAT problems.
This motivates our effort to improve the first step of the inference problem to obtain simpler SAT instances.
While previous attempts have aimed at reducing the size of APTAs~\citep{DBLP:conf/icgi/LangPP98,Oncina-Garcia-92,DBLP:journals/pr/BugalhoO05}, we introduce a new and incremental construction of the APTAs that comes with a \emph{minimality} guarantee for the acceptor of the given samples.

\paragraph*{Contributions.}
We propose employing the (polynomial-time) incremental construction of minimal acyclic DFA learning algorithm~\citep{DBLP:journals/coling/DaciukMWW00} for minimal DFAs from a given set of \emph{positive} samples.
    We offer two constructions based on it. The first consists of building two minimal DFAs, $\D^+$ and $\D^-$, for the positive samples $S^+$ and the negative ones $S^-$, respectively.
When composing them, we set as rejecting the accepting states of $\D^-$ and use the DFA pair $(\D^+, \D^-)$ as the acceptor of $S = (S^+, S^-)$.
For the second construction, our algorithm directly learns an APTA from $S$, hence considering both sets of positive and negative samples at the same time.
As a consequence, we extend the algorithm to support APTA learning from the pair of labelled samples $S$.
The obtained APTA is guaranteed to be the \emph{minimum-size deterministic} acceptor for $S$.

We have implemented these techniques in our new tool \tool~and compared it with the state-of-the-art tools \dfaind~\citep{DBLP:conf/lata/UlyantsevZS15,DBLP:conf/lata/ZakirzyanovMIUM19} and \dfaid~\citep{DBLP:conf/fmcad/LaufferYVSS22}, on the benchmarks generated as described in~\citep{DBLP:conf/lata/UlyantsevZS15,DBLP:conf/lata/ZakirzyanovMIUM19}.
Our experimental results demonstrate that \tool~builds smaller APTAs and is therefore significantly faster at minimising the DFAs than both \dfaind~and \dfaid.

To test the limitation of our technique, we employed it to extract deterministic safety or reachability automata as witness automata for parity game solving.
With \tool, we have established the lower bounds on the size of deterministic safety automata for parity games with up to $7$ colours.
In this case, the main bottleneck is no longer solving the Min-DFA inference problem, but the generation of the labelled samples, whose number is exponential in the length and the number of colours.
To the best of our knowledge, this is the first time that Min-DFA inference tools have been applied to parity game solving.
If they eventually scale, this may lead to new insights into the actual size of the minimal safety automata for solving parity games.

\section{Preliminaries}

In the whole paper, we fix a finite \emph{alphabet} $\alphabet$ of letters.
A \emph{word} is a finite sequence of letters in $\alphabet$. We denote with $\emptyword$ the empty word and with $\finwords$ the set of all finite words.
When discarding the empty word, we restrict the set of words to $\poswords = \finwords\setminus\setnocond{\emptyword}$.
A subset of $\finwords$ is a \emph{finitary language}.
Given a word $u$, we denote by $\wordletter{u}{i}$ the $i$-th letter of $u$.
For two given words $u$ and $w$, we denote by $u \cdot w$ ($uw$, for short) the concatenation of $u$ and $w$.
We say that $u$ is a \emph{prefix} of $u'$, denoted as $u \preceq u'$, if $u' = u\cdot v$ for some word $v\in\finwords$.
We denote by $\pref(u)$ the set of the prefixes of $u$, i.e., $\setcond{v \in \finwords}{ v \preceq u}$.
We also extend function $\pref$ to sets of words, i.e., we have $\pref(S) = \bigcup_{u \in S} \pref(u)$. 
 
\paragraph{Automata.}
An automaton on finite words is called a 3-valued (deterministic) \emph{finite automaton} (3DFA).
A 3DFA $\A$ is formally defined as a tuple $(\states, \init, \trans, F, R)$, where $\states$ is a finite set of states, $\init \in \states$ is the initial state, $\trans: \states \times \alphabet \rightarrow \states$ is a transition function, and $F, R$ and $D = \states\setminus (F\cup R)$ form a partition of $\states$ where $F\subseteq \states$ is the set of \emph{accepting} states, $R \subseteq \states$ is the set of \emph{rejecting} states and $D$ is the set of \emph{don't-care} states. 3DFAs map all words in $\finwords$ to three values, i.e., accepting ($+$), rejecting ($-$) and don't-care ($?$). This is why we call it a 3DFA.

A \emph{run} of a 3DFA $\A$ on a not empty finite word $u$ of length $n$ is a sequence of states $\run = q_{0} q_{1} \cdots q_{n} \in \states^{+}$ such that, for every $0 \leq i < n$, $q_{i+1} \in \trans(q_{i}, \wordletter{u}{i})$.
As usual, a finite word $u \in \finwords$ is \emph{accepted} (respectively, \emph{rejected}) by a 3DFA $\A$ if there is a run $q_{0} \cdots q_{n}$ over $u$ such that $q_{n} \in F$ (respectively, $q_n \in R$).
Naturally, a 3DFA $\A$ can be seen as a classification function in $\finwords \rightarrow \setnocond{+, -, ?}$.
The usual deterministic finite automaton (DFA) is a 3DFA with $R = \states\setminus F$, i.e., a word is mapped to either $+$ or $-$.
Both the classes of words \emph{accepted} and \emph{rejected} by 3DFAs are known to be regular languages.

We remark that the 3DFAs are very standard model for representing positive and negative samples in the literature.
In~\citep{alquezar1995incremental}, 3DFAs are called deterministic unbiased finite state automata.

For a given regular language, the Myhill-Nerode theorem~\citep{Myhill57,Nerode58} helps to obtain the minimal DFA.
Similarly, it is suggested in~\citep{DBLP:conf/tacas/ChenFCTW09} that we can identify equivalent words that reach the same state in the minimal 3DFA of a given function $L: \finwords \rightarrow \setnocond{+, -, ?}$.
Let $x, y$ be two word in $\finwords$ and $L \in (\finwords \rightarrow \setnocond{+, -, ?})$ be a function. We define an equivalence relation $\canoEq_L \subseteq \finwords \times \finwords$ as below:
\[ x \canoEq_L y \text{ if, and only if, } \forall v \in \finwords, L(xv) = L(yv).\]
We denote by $\size{\canoEq_L}$ the index of $\canoEq_L$, i.e., the number of equivalence classes defined by $L$.
Let $S$ be a given finite set of labelled samples.
We can also see $S$ as a classification function and it also induces an equivalence relation $\canoEq_S$.
Further, we define $\posset = \setnocond{ u \in S: S(u) = +}$, $\negset = \setnocond{u \in S: S(u) = -}$ and $S^? = \finwords \setminus S = \setnocond{u \in S: S(u) = ?}$.
Finally, we conclude with a straightforward proposition that follows from the fact that $\size{\canoEq_S}$ is bounded by the number of prefixes of $S$, i.e., $\size{\pref(S)}$.
\begin{fact}
Let $S$ be a finite set of labelled samples. Then the index of $\canoEq_S$ is also finite.
\end{fact} 

\section{Overview of \tool}

\paragraph*{Problem definition.}
Let $S = (S^+, S^-)$ be the given set of labelled samples.
Our goal in this paper is to find a \emph{minimal} DFA (Min-DFA) $\D$ for $S$ such that for all $u \in \finwords$, if $S(u) = \$$, where $\$ \in \setnocond{+, -}$, then $\D(u) = \$$.
We call the target DFA a minimal \emph{separating} DFA\footnote{In \citep{DBLP:conf/tacas/ChenFCTW09}, the 3DFA that recognises $S$ is called separating DFA for $S$.} for $S$.

The passive learner for separating DFAs~\citep{DBLP:conf/icgi/HeuleV10,DBLP:conf/lata/ZakirzyanovMIUM19} usually works as follows:
(1) Construct a 3DFA $\M$ recognising $S$ and (2) Minimise the 3DFA $\M$ to a Min-DFA using a SAT solver.
Our work mainly differs from the APTA construction for the first step.
We will first describe the APTA construction in Section~\ref{ssec:prior-apta-construction} and then give our proposal for the tool architecture in Section~\ref{ssec:our-apta-construction}. 
In the remainder of the paper, we let $S = (S^+, S^-)$ be the given labelled sample set.

\subsection{The construction of 3DFAs}
\label{ssec:prior-apta-construction}
Prior works~\citep{DBLP:conf/lata/ZakirzyanovMIUM19,DBLP:conf/lata/UlyantsevZS15,DBLP:conf/icgi/HeuleV10} construct an automaton called \emph{augmented prefix tree acceptor} (APTA)~\citep{alquezar1995incremental,DBLP:conf/icgi/CosteN98,DBLP:conf/ecml/CosteF03}\footnote{APTAs are called prefix tree unbiased finite state automata in~\citep{alquezar1995incremental} and they are also similar to the prefix-tree Moore Machines in \citep{trakhtenbrot1973finite}.} $\P$ that recognises $S$.
The APTA $\P = (\states, \emptyword, \trans, F, R)$ is formally defined as a 3DFA where $\states = \pref(S)$ is the set of states, $\emptyword$ is the initial state, $F = \posset$ is the set of accepting states, $R = \negset$ is the set of rejecting states, and $\trans(u, a) = ua$ for all $u, ua \in \states$ and $a \in \alphabet$.
As mentioned in the introduction, the number of boolean variables and clauses used in the SAT problem is polynomial in the number of states of $\P$.
The main issue is that the size of $\P$ increases dramatically with the growth of the number of samples in $S$ and the length of the samples, This is not surprising, given that $\P$ maps every word in $\pref(S)$ to a unique state.
To show this growth, we considered samples from parity game solving. Table~\ref{tab:3DFA-apta} shows the size comparison between the APTA and its minimal 3DFA (Min-3DFA) representation.
With 5 and 6 letters (in this case colours), we can observe that the Min-3DFAs can be much smaller than their corresponding APTA counterparts.
In other words, there are a lot of equivalent and redundant states in APTAs that can be merged together.

\begin{table}
    \centering
\begin{tabular}{ c c | c  c }\toprule  
    $\size{\alphabet}$ & Length & Min-3DFA & APTA \\ \hline
    5 & 7 & 438 & 53,277\\
    5 & 8 & 541  & 209,721\\
    5 & 9 & 644 & 835,954\\
    5 & 10 & 747 & 3,369,694\\
\hline
    6 & 7 & 1279 & 199,397\\ 
    6 & 8 & 1807 & 930,870\\ 
    6 & 9 & 2170 & 4,369,362 \\
    6 & 10 & 2533 & 20,689,546\\ 
    \bottomrule
\end{tabular}
\caption{Size comparison between Min-3DFA and APTA on part of benchmarks for parity game solving.}\label{tab:3DFA-apta}
\end{table}

In fact, since APTAs are acyclic (i.e., there are no cycles), we can minimise them with a linear-time backward traversal~\citep{DBLP:journals/coling/DaciukMWW00}.
The crucial step is how to identify whether two states are equivalent in the backward traversal~\citep{DBLP:journals/coling/DaciukMWW00}.
Our solution is to use the equivalence relation $\canoEq_S$.
Based on the definition of $\canoEq_S$, we define that two states $p, q \in \states$ are equivalent, denoted $p \equiv r$ if, and only if,
\begin{enumerate}
    \item they have the same acceptance status, i.e., they are both accepting, rejecting or don't-care states;
    \item for each letter $a \in \alphabet$, either they both have no successors or their successors are equivalent.
\end{enumerate}
In the implementation, since we only store one representative state for each equivalence class, the second requirement can be changed as follows:
\begin{enumerate}
   \item[2'.] for each letter $a \in \alphabet$, either they both have no successors or the same successor.
\end{enumerate}

Therefore, it is easy to come up with an algorithm to minimise the given APTA tree $\P$ by doing the following:
\begin{enumerate}
    \item We first collapse all accepting (respectively, rejecting) states without outgoing transitions to one accepting (respectively, rejecting) state without outgoing transitions, and put the two states in a map called $Register$ that stores equivalence relation of states. 
    \item Then we perform \emph{backward} traversal of states and check if there is a state whose successors are \emph{all} in $Register$. For such states, we identify equivalent states by Rule 2', replace all equivalent states with their representative and put their representative in $Register$.
    \item We repeat step 2 until all states, including the initial state, are in $Register$.
\end{enumerate}
In this way, we are guaranteed to obtain the minimal 3DFA $\M$ that correctly recognises the given set $S$.

Further, we do not have to construct the APTA $\P$ in order to obtain the minimal 3DFA for the given samples.
We have shown that it is possible to construct the minimal 3DFA on the fly, given that the samples are in the usual lexicographic order.
If the samples are not in the lexicographic order, we will first order them and then proceed to the construction of its 3DFA.
The details of the construction are deferred to Appendix~\ref{apd:min-apta-constrution}.

\begin{theorem}
    The incremental construction produces the minimal 3DFA recognising exactly the sample set $S$.
\end{theorem}

\subsection{Our proposal}
\label{ssec:our-apta-construction}
Our tool \tool~follows the classic two-step workflow. 
The advantage of \tool~over prior works is that \tool~has an access to an \emph{incremental} construction that produces the minimal 3DFA $\M$ of $S$ with respect to $\canoEq_S$.
With the incremental construction, a natural workflow of \tool~is to first construct the minimal 3DFA $\M$ from $S = (S^+, S^-)$ and then find a minimal separating DFA $\D$ from $\M$ using a SAT solver.
This approach is depicted in Figure~\ref{fig:work-flow}.
All the components labelled in green or blue are novel contributions made in our tool.
We use the standard SAT-based minimal separating DFA extractions from APTAs (and hence 3DFAs) described in \citep{DBLP:conf/lata/UlyantsevZS15}.

We observe that the minimal separating DFA finding algorithm~\citep{DBLP:conf/lata/UlyantsevZS15} does not necessarily work only on 3DFAs, but also on multiple 3DFAs; see Appendix~\ref{apd:sat-encoding} for our formulation.
This motivates us to ask a question: can we construct multiple 3DFAs for $S$?
We give a affirmative answer to this question.
\begin{figure}[t]
    \centering
    \scalebox{0.85}{
    \begin{tikzpicture}
        
        \node[draw=black, anchor=north, diamond,  text width=1.4cm,fill=black!20,minimum height=1.cm] (Mem) at (-0.1,-0.278) {};
        \node[] at (-0.1, -1.2) {$\posset, \negset$};

       \node[draw=green!50, inner sep=0, anchor=north,rounded corners, text width=2.8cm,fill=green!10,minimum height=0.9cm] (tree_learner) at (3.76, -0.76) {};
       \node[] at (3.76, -1.2) {Min-3DFA $\M$};
        
       \node[draw=yellow!70, inner sep=0, anchor=north,rounded corners, text width=2.6cm,fill=yellow!10,minimum height=1.4cm, align = center] (table_learner) at (7.86, -0.5) {
        SAT Encoding \\ SAT Solving
       };
	\node[align = center, text width=2.6cm] at (7.86, -0.22) {
			Minimiser
		};

        \node[draw=black, anchor=north,  text width=1.5cm,fill=black!20,minimum height=0.8cm] (minDFA) at (11.52, -0.76) {};
        \node[align = center, text width=2.6cm] at (11.52,-1.2) {
			Min-DFA
		};

        \draw [line width=1pt, double distance=1pt,
             arrows = {-Latex[length=0pt 3 0]}] (0.9, -1.2) -- node[auto] {}(2.3, -1.2) ;

        \draw [line width=1pt, double distance=1pt,
             arrows = {-Latex[length=0pt 3 0]}] (5.2, -1.2) -- node[auto] {}(6.5, -1.2) ;

        \draw [line width=1pt, double distance=1pt,
             arrows = {-Latex[length=0pt 3 0]}] (9.23, -1.2) -- node[auto] {}(10.6, -1.2) ;
	
\end{tikzpicture}
    }
    \caption{Workflow of \tool~with 3DFAs}
    \label{fig:work-flow}
\end{figure}
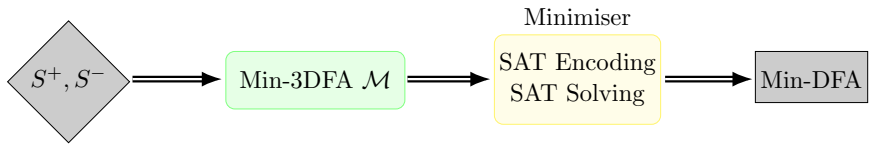

Our construction of double 3DFAs (dDFAs) from $S=(\posset, \negset)$ is formalised as follows:
(1) Construct the minimal 3DFAs $\D^+$ and $\D^-$ for $(\posset, \emptyset)$ and $(\negset, \emptyset)$, respectively;
(2) Make sure that $\D^+$ and $\D^-$ do not share the same state names;
(3) Combine the two 3DFAs into a dDFA $\N$ where the initial states of $\S$ are the initial states of both $\D^+$ and $\D^-$, the transitions between states remain unchanged and we make the accepting states of $\D^-$ as rejecting states.
The workflow of this construction is depicted in Figure~\ref{fig:work-flow-3NFAs}.
In this way, we obtain a dDFA $\N$ that recognises exactly the given set $S$.
The empirical evaluation shows that the two types of workflows are incomparable and both have their place in the learning procedure.

The SAT encoding of the minimiser component is fairly standard and has been deferred to Appendix~\ref{apd:sat-encoding}.
We can then gradually increase the number of states in the proposed separating DFA for $S$ until the constructed SAT formula is satisfiable.
It immediately follows that:
\begin{theorem}
    Our tool {\tool} will output a minimal separating DFA for $S$.
\end{theorem}

We remark that in~\citep{alquezar1995incremental}, non-incremental and incremental constructions were also proposed to find small and even minimal 3DFAs $M$ such that for each $u \in S^{\$}$, it holds that $M(u) = \$$ where $\$ \in \setnocond{+, -}$.
These two constructions are based on state merging techniques given in \citep{Oncina-Garcia-92}, with the worst-time complexity being $\O(\size{\pref(S)}^3)$, while the size of the intermediate 3DFA constructed by our incremental construction will not exceed $\size{\pref(S)}$.
Furthermore, the resultant 3DFA $M$ may accept (respectively, reject) more words than $S^+$ (respectively, $S^-$), while our incremental construction produces a minimal 3DFA recognising $S$ exactly.
As a consequence, their constructed 3DFA can be smaller (or even larger) than ours, and can no longer be used to extract the minimal separating DFA for $S$. 

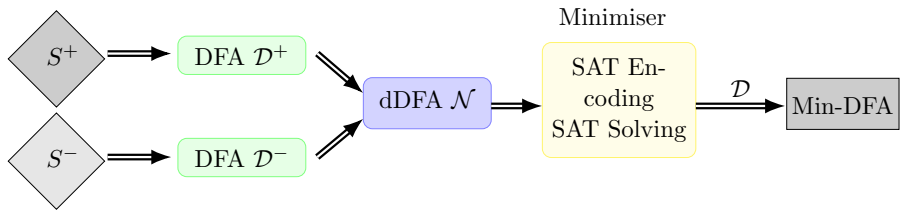
\begin{figure}[t]
    \centering
    \scalebox{0.85}{
    \begin{tikzpicture}
        
        \node[draw=black, anchor=north, diamond,  text width=1cm,fill=black!20, minimum height=0.5cm] (MemPos) at (-0.3,0.3) {};
        \node[] at (-0.2, -0.4) {$\posset$};

        \node[draw=black, anchor=north, diamond,  text width=1cm,fill=gray!20, minimum height=0.5cm] (MemNeg) at (-0.3,-1.3) {};
        \node[] at (-0.2, -2) {$\negset$};

      \node[draw=blue!50, inner sep=0, anchor=north,rounded corners, fill opacity=0.85,text width=2cm,fill=blue!20,minimum height=0.8cm] (learner) at (5.5, -0.75) {};
      \node[] at (5.5, -1.1) {dDFA $\N$};

       \node[draw=green!50, inner sep=0, anchor=north,rounded corners, text width=2cm,fill=green!10,minimum height=0.6cm] (tree_learner) at (2.6, -0.1) {};
       \node[] at (2.6, -0.4) {DFA $\D^+$};

       \node[draw=green!50, inner sep=0, anchor=north,rounded corners, text width=2cm,fill=green!10,minimum height=0.6cm] (tree_learner) at (2.6,-1.7) {};
       \node[] at (2.6, -2.0) {
       DFA $\D^-$	
       };

       \node[draw=yellow!70, inner sep=0, anchor=north,rounded corners, text width=2.4cm,fill=yellow!10,minimum height=1.8cm, align = center] (table_learner) at (8.5, -0.2) {
        SAT Encoding \\ SAT Solving	
       };
		\node[align = center, text width=2.6cm] at (8.4, 0.2) {
			Minimiser
		};

        \node[draw=black, anchor=north,  text width=1.5cm,fill=black!20,minimum height=0.8cm] (minDFA) at (12, -0.76) {};
        \node[align = center, text width=2.6cm] at (12,-1.2) {
			Min-DFA
		};

        \draw [line width=1pt, double distance=1pt,
             arrows = {-Latex[length=0pt 3 0]}] (0.5, -0.38) -- node[auto] {}(1.5, -0.38) ;

         \draw [line width=1pt, double distance=1pt,
             arrows = {-Latex[length=0pt 3 0]}] (0.5, -2.0) -- node[auto] {}(1.5, -2.0) ;

        \draw [line width=1pt, double distance=1pt,
             arrows = {-Latex[length=0pt 3 0]}] (3.78, -0.38) -- node[auto] {}(4.5, -1.0) ;

        \draw [line width=1pt, double distance=1pt,
             arrows = {-Latex[length=0pt 3 0]}] (3.78, -2) -- node[auto] {}(4.5, -1.4) ;

        \draw [line width=1pt, double distance=1pt,
             arrows = {-Latex[length=0pt 3 0]}] (6.5, -1.2) -- node[auto] {}(7.3, -1.2) ;

        \draw [line width=1pt, double distance=1pt,
             arrows = {-Latex[length=0pt 3 0]}] (9.7, -1.2) -- node[auto] {$\mathcal D$}(11.1, -1.2) ;
	
\end{tikzpicture}
    }
    \caption{Workflow of \tool~with dDFAs}
    \label{fig:work-flow-3NFAs}
\end{figure}

\section{Evaluation}
\label{sec:eval}
To further demonstrate the improvements of \tool\footnote{\tool~is publicly available at \url{https://github.com/liyong31/DFAMiner}}~over the state of the art, we conducted comprehensive experiments on standard benchmarks~\citep{DBLP:conf/lata/UlyantsevZS15,DBLP:conf/lata/ZakirzyanovMIUM19}.
We compared with \dfaind~\citep{DBLP:conf/lata/ZakirzyanovMIUM19} and \dfaid~\footnote{\url{https://github.com/mvcisback/dfa-identify}}~\citep{DBLP:conf/fmcad/LaufferYVSS22}, the state of the art tools publicly available for passive learning tasks.
Unlike \tool~and \dfaind, \dfaid~uses a SAT encoding of graph coloring problems~\citep{DBLP:conf/icgi/HeuleV10} and the representative DFAs in the second step~\citep{DBLP:conf/lata/UlyantsevZS15}. 
Like \dfaind, \tool~is also implemented in Python with \textsf{PySAT}~\citep{DBLP:conf/sat/IgnatievMM18}.
We delegate all SAT queries to the SAT solver \textsf{CaDical}~1.5.3 in all tools~\citep{BiereFazekasFleuryHeisinger-SAT-Competition-2020-solvers}.
\tool~accepts the samples formalised in Abbadingo format\footnote{\url{https://abbadingo.cs.nuim.ie/}}.

\begin{table}
\centering
\begin{tabular}{ c | c c | c c | c c | c c }\toprule
    \multicolumn{1}{c}{} & \multicolumn{2}{c}{DFA-Inductor} & \multicolumn{2}{c}{DFA-Identify} & \multicolumn{2}{c}{dDFA-MIN} & \multicolumn{2}{c}{3DFA-MIN} \\
    \hline
    N & avg & \% & avg & \% & avg & \% & avg & \% \\
    \hline
    4 & 0.12 & 100 & 0.09 & 100 & 0.03 & 100 & 0.02 & 100 \\
    5 & 0.29 & 100 & 1.38 & 100 & 0.06 & 100 & 0.05 & 100 \\
    6 & 0.67 & 100 & 2.33 & 100 & 0.30 & 100 & 0.18 & 100 \\
    7 & 1.81 & 100 & 4.12 & 100 & 0.80 & 100 & 0.73 & 100 \\
    8 & 3.57 & 100 & 9.70 & 100 & 1.29 & 100 & 1.25 & 100 \\
    9 & 10.84 & 100 & 20.76 & 100 & 3.83 & 100 & 3.78 & 100 \\
    10 & 50.91 & 100 & 44.57 & 100 & 17.88 & 100 & 16.80 & 100 \\
    11 & 154.73 & 100 & 128.69 & 100 & 55.12 & 100 &  59.46 & 100 \\
    12 & 399.52 & 96 & 373.65 & 99 & 144.27 & 100 & 162.39 & 100 \\
    13 & 850.04 & 74 & 785.93 & 82 & 390.10 & 99 & 418.62 & 97 \\
    14 & 1125.59 & 19 & 1099.92 & 23 & 809.88 & 76 &  861.10 & 69 \\
    15 & 1182.98 & 6 & 1197.61 & 1 & 1060.18 & 37 & 1062.02 & 34 \\
    16 & 1188.17 & 1 & 1184.82 & 3 & 1167.58 & 4 & 1164.02 & 5 \\
    \bottomrule
\end{tabular}
\caption{Comparison for the minimisation of DFAs from random samples of \tool~with DFA inductor. For each approach we report the mean minimisation time and the percentage of DFAs minimised within the time limit.}\label{tab:exp-short}
\end{table}

The experiments of Table~\ref{tab:exp-short} were carried on an Intel i7-4790 3.60 GHz processor. Each index N, reports the results of 100 benchmark instances of random samples. Each benchmark has $50\times N$ samples. For every index, we show the average time and the percentage of instances solved within 1200 seconds.
The alphabet for the samples has two symbols while the size of the generated DFA is $N$.
We compare four approaches to inferring Min-DFAs: \dfaind, \dfaid, \tool~with 3DFA (DFA-MIN), and \tool~with double DFA (dDFA-MIN). An extended version of the table is reported in Appendix~A.3.
Both dDFA-MIN and 3DFA-MIN perform better than \dfaind~and \dfaid, on average they are three times faster.
\dfaind~can minimise within 20 minutes instances up to level 13, while the two variants of \tool~can scale one more level and minimise one third of the instances of level 15. On these random samples the double DFA approach is slightly faster than the 3DFA one.

Figures~\ref{fig:test1} and~\ref{fig:test2} reports the comparison on the size of the APTA/dDFA (on the left) and minimisation time (on the right) for the previous benchmark. 
In these two figures, instead of the mean data, we show the individual data of each sample.
Both \dfaind~and \dfaid~build the same APTA (they differ for the encoding step), and as shown in Figure~\ref{fig:test1}, its size is three times larger than the dDFA built by \tool, no matter how big is the final DFA. Figure~\ref{fig:test2}, instead, shows that when using a double DFA, \tool~always performs better than DFA Inductor, on average three times faster with peaks of more than four times faster. We provide additional details and comparisons on minimisation time and 3DFA size in Appendix A.3.

\begin{figure}
    \centering
\begin{minipage}{.478\textwidth}
  \centering
  \includegraphics[width=1\linewidth]{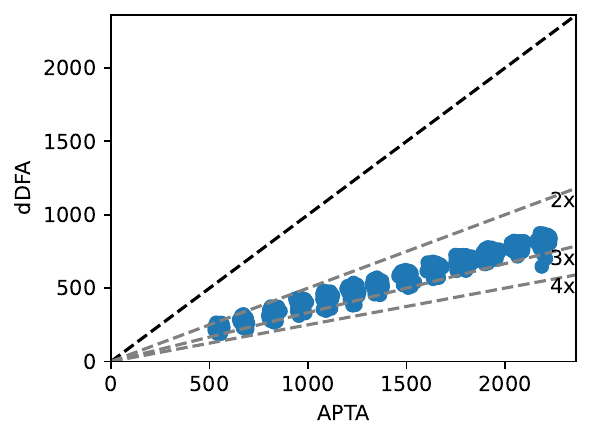}
  \caption{figure}{Scatter plot on automata size}
  \label{fig:test1}
\end{minipage}%
\begin{minipage}{.5\textwidth}
  \centering
  \includegraphics[width=1\linewidth]{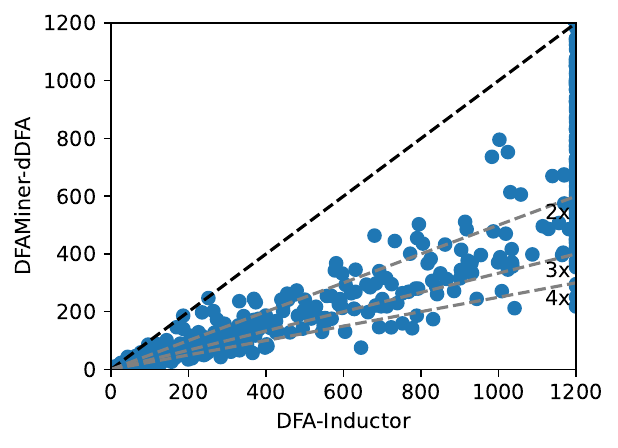}
  \caption{figure}{Scatter plot on runtime (secs)}
  \label{fig:test2}
\end{minipage}
\end{figure}

\section{Application in parity game solving}
\label{sec:application}

It has been shown that the first quasi-polynomial parity game solving algorithm~\citep{CJKLS17} essentially builds a separating automaton of quasi-polynomial size to distinguish runs with only winning cycles (cycles closed in a parity game where the highest colour occurring is \emph{even}) from runs that contain only odd cycles (where the highest colour occurring is \emph{odd})~\citep{CDFJLP19}.

Without going into detail, we note that the hardest case occurs when the colours are unique (occur only once, hence, the colour identifies a node, thing that helps to detect cycles), and have implemented this as follows: we fix an alphabet with \textbf{c} different colours, a length $\ell > \textbf{c}$, and a highest colour \textbf{c}.
We must accept a word if all cycles are winning (e.g.\ 001212), rejecting it if all cycles are losing (e.g.\ 13123312). Words with winning and losing cycles (e.g.\ 21232) are don't-care words.
A cycle occurs when a colour has repeated at least twice.

The resulting automata are always safety that reject all words that have not seen a winning cycle after (at most) $\ell$ steps, as well as some words that have seen both, winning and losing cycles (don't-care word), or, alternatively, reachability automata that accept all words that have not seen a losing cycle after at most $\ell$ steps (again, except don't-care ones).
Thus, the size of the Min-DFA falls when increasing the sample length $\ell$, and \emph{eventually stabilises}.
Using such a separating automaton reduces solving the parity game to solving a safety game \citep{BC18}.

Separating automata build with the current state-of-the-art construction~\citep{CJKLS17} grow quasi-polynomially, and since it is not known whether these constructions are optimal, we applied \tool~to learn the most succinct separating automata for the parity condition.

\begin{table}
\centering
\begin{tabular}{| c || c | c | c | c | c |}
    \hline
    Colours & 2 & 3 & 4 & 5 & 6 \\
    \hline
    DFA Size & 3 & 3 & 5 & 5 & 9 \\
    \hline
    Length & 3 & 5 & 7 & 11 & 15 \\
    \hline
    \#Pos & 3 & 130 & 1,645 & 9,375,269 & 4,399,883,736 \\
    \#Neg & 5 & 31 & 5,235 & 1,009,941 & 38,871,920,470 \\
    \hline
\end{tabular}
\caption{Samples required to learn the minimal separating automata for solving parity games.}\label{tab:parity}
\end{table}

Table~\ref{tab:parity} shows the application of \tool~to the parity condition up to $7$ colours (from 0 to 6).
For each maximal colour we report the length required to build the minimal separating automaton, the size of the obtained DFA, and the number of all the positive and negative samples generated. 
Although most words have both wining and losing cycles (don't-care words), the positive and negative samples grow \emph{exponentially}, too, which is why we stopped at 7 colours.

While the APTA size constructed by \dfaind~grows exponentially, the sizes of dDFAs and 3DFAs seem to grow only constantly when increasing the length of the samples for a fixed colour number. We report the details in Table~\ref{tab:parity-apta-3dfa} in Appendix~\ref{app:apta-3dfa-size}.
Consequently, \emph{all} versions of \dfaind~were only able to solve cases with at most 4 colours, while \tool~can manage to solve cases up to 6 colours and length 16.
To further push the limit of \tool~for parity game solving, we have also provided an efficient SAT encoding for parity games (in Appendix~\ref{app:parity}).
With the constructions for both 3DFAs and dDFAs and the efficient encoding, the bottleneck of the whole procedure is no longer solving the Min-DFA inference problem, but the generation of samples.
With a better sample generation approach, we believe that this application can give insights on the structure of minimal safety automata for an arbitrary number of colours.

\section{Discussion and Future Work}

We propose a novel and more efficient way to build APTAs for the Min-DFA inference problem.
Our contribution focuses on a compact representation of the positive and negative samples and, therefore, provides the leeway to benefit from further enhancements in solving the encoded SAT problem.

Natural future extensions of our approach include implementing the tight encoding of symmetry breaking~\citep{DBLP:conf/lata/ZakirzyanovMIUM19}.
Another easy extension of our construction is to learn a set of decomposed DFAs~\citep{DBLP:conf/fmcad/LaufferYVSS22}, thus improving the overall performance as well.
A more complex future work is to investigate whether or not one can similarly construct a deterministic B\"uchi automaton based on $\omega$-regular sets of accepting, rejecting, and don't-care words that provides a minimality guarantee for a given set of labelled samples.

\bibliography{paper}

\appendix

\section{Minimal APTA construction}\label{apd:min-apta-constrution}

The algorithm can be seen as the combination of the construction of the APTA and its minimisation based on the backward traversal of the APTA tree.
We will enumerate the input samples $S$ one by one in their lexicographical order;
this is important for identifying the states/nodes in the APTA tree that have been completely traversed and will remain unchanged after some point, in order to obtain an on-the-fly minimisation.
We will highlight the two fundamental components: the minimisation and the backward traversal/construction of the APTA tree in the sequel.

Let us consider a simpler situation where the \emph{full} APTA tree $\P$ is already given.
The most important thing in the minimisation component is to decide whether two states $p$ and $q$ are equivalent and it has already been described.

By using the equivalence relation we can obtain the minimal 3DFA $\M$ that correctly recognises the given set $S$.
Moreover, if we use a hash map for storing all representative states in $Register$, the minimisation algorithm above runs in linear time with respect to the number of states in $\P$.
However, as we can see in Table~\ref{tab:3DFA-apta}, the APTAs can be \emph{significantly larger} than the corresponding minimal 3DFAs.
Hence, it is vital to avoid the full construction of the APTA tree $\P$ of $S$.
The key to the on-the-fly construction is to know when a state has been \emph{completely} traversed during construction.

To this end, we can just assume that the samples are already ordered in the usual lexicographical order.
The comparison works as follows.
Assume that we already have a default order on the letters in $\alphabet$.
For two words $u$ and $u'$, we first compare their first $\mathsf{min}(\size{u}, \size{u'})$ letters:
(1) if we find a word that has a smaller letter than the other at the same position, then that word is smaller,
(2) if all letters are the same and $u$ has the same length as $u'$, then they are equal,
otherwise (3) the word that is longer than the other word is greater.

Assume that $S = u_1, u_2, \cdots, u_{\ell}$ is in order. We describe below how to tell when a state is \emph{impossible} to have more successors and it is ready to find its representative state. 
Assume that current APTA is $\P_i = (\states_i, \setnocond{\init}, \trans_i, F_i, R_i)$ and we now input the next sample $u_{i+1}$, where $i \geq 0$.

When $i = 0$, $\P_0$ is of course minimal since $\P_0$ only has a state $\init$ without any outgoing transitions.
For technical reason, we let $u_{0} = \emptyword$, which may not appear in the sample set $S$.
(Note that if there is an empty word $\emptyword$ in $S$, $\init$ will be set to accepting or rejecting accordingly.)

We first run $u_{i+1}$ on $\P_i$.
We let $u_{i+1} = x\cdot y_{i+1}$ and assume that $x \in \pref(u_{i+1})$ be the \emph{longest} word such that $\trans_i(\init, x) \neq \emptyset$.
Let $p = \trans_i(\init, x)$.
Then, all states along the run of $\P_i$ over $x$ are not ready to find their representatives, as $p$ needs to add more reachable states to run the suffix $y_{i+1}$.
Moreover, we observe that $x$ must be a prefix of $u_i$, i.e., $x \in \pref(u_i)$.
This is because every word that has a \emph{complete} run in $\P_i$ must \emph{not} be greater than $u_i$.
By definition of the lexicographic order, if $x$ is not a prefix of $u_i$, then $x$ must be smaller than $\wordletter{u_i}{0\cdots \size{x}}$.
This leads to the contradiction that $u_i$ is greater than $u_{i+1}$.
Let $ u_i = x\cdot y_{i}$ and $\rho = p_{0} \cdots p_{\size{u_i}}$ where $p_0 = \init$ and $p_{\size{x}} = p$.
We can show that all the states $p_k$ with $k > \size{x}$ in the run of $\P_i$ over $u_i = x \cdot y_i$ are ready to merge with their representative, as they must not have more reachable states.
Assume that there is a state $p_{\ell}$ with $\ell > \size{x}$ reached over a sample $u_h$ with $h > i$, then it is easier to lead the contradiction that $u_h$ is smaller than $u_{i+1}$.
Therefore, we can similarly identify the representatives for those states and merge them in the usual backward manner.
It follows that all the states except the ones in the run of $u_{i+1}$ in the 3DFA $\P_{i+1}$ are already consistent with respect to $\canoEq_S$; thus, there is no need to modify them afterwards.
After we have input all the samples, we only need to merge all states in the run over $u_{\ell}$ with their equivalent states.
In this way, we are guaranteed to obtain the minimal 3DFA $\S$ for $S$ in the end.

The formal procedure of the above incremental construction of the minimal 3DFA from $S$ is given in Algorithm~\ref{alg:the-algorithm-for-3DFA}.
Note that, when looking for the run from $p$ over the \emph{last} input sample, we only need to find the successors over the maximal letter by the $\lastchild$ function. 
In this way, when we reach the last state of the run, we then can begin to identify equivalent states in a backward manner, as described in the subprocedure $\register$.
Moreover, in the function $\mathsf{add\_suffix}(p, \var{y})$, we just create the run from $p$ over $y$ and set the last state to be accepting or rejecting depending on the label of $u$.
In fact, we only extend the equivalence relation $\equiv$ of~\citep{DBLP:journals/coling/DaciukMWW00} in $\register$ to support the accepting, rejecting and don't-care states, as described before.

\SetAlgoNoLine
\begin{algorithm}
\caption{Incremental construction of the minimal 3DFA from $S$}\label{alg:the-algorithm-for-3DFA}
    \begin{algorithmic}
    \Procedure{Main procedure}{Sample Set $U$}
        \State $\var{Register} := \emptyset$
        \State \While{$U$ has next sample $u$}{
        \State $\var{x} := \mathsf{common\_prefix}(u)$
        \State $\var{p} := \trans(\init, \var{x})$ \Comment{the last state over the common prefix $x$}
        \State $\var{y} := \wordletter{u}{\size{x} \ldots}$ \Comment{the remaining suffix of $u$}
        \State \If{$\haschild(p)$}{
        \State $\register(p)$ \Comment{merge/register all states after $p$}
         }
        \State $\mathsf{add\_suffix}(p, \var{y})$ \Comment{create run to accept suffix $y$ from $p$}
        }
        \State $\register(\init)$  \Comment{merge the run over the last sample}
        \EndProcedure
        
        \Procedure{$\register(p)$}{}
        \State $r := \lastchild(p)$  \Comment{obtain the successor over the maximal letter}
        \State \If{$\haschild(r)$}{
        \State $\register(r)$ \Comment{recursively obtain the run over last sample}
        }
        \If{$\exists q \in \states. (q \in Register \land q \equiv r)$}{
        \State\ $\mathsf{last\_child}(p) := q$ \Comment{merge with its representative}
        }
        \Else {
        \State\ $Register := Register \cup \{r\}$ \Comment{set the first state of each class as representative}
        }
    \EndProcedure
    \end{algorithmic}
\end{algorithm}

\begin{theorem}
Let $S$ be a finite labelled set of ordered samples.
Algorithm~\ref{alg:the-algorithm-for-3DFA} returns the correct and minimal 3DFA recognising $S$.
\end{theorem}
\begin{proof}
The proof is basically an induction on the number of input samples and has been overlapped with the intuition described above.
We thus omit it here.
\end{proof}

\section{SAT Encoding}\label{apd:sat-encoding}
In this section, we show that how to obtain the minimal separating DFAs from 3DFAs/3NFAs.
It is known that minimising DFAs with don't-care words is NP-complete~\citep{DBLP:journals/tc/Pfleeger73}.
We will take the advantage of the current powerful solvers for Boolean Satisfiability (SAT) problems to look for minimal DFAs.

\subsection{SAT-based encoding of minimisation}
We assume that we are given a dDFA $\N = (\T, A, R)$, where $\T = (\states, \inits, \trans)$ is obtained from two DFAs $\D^+$ and $\D^-$.
We look for a separating DFA $\D$ of $n$ states for $\N$ such that for each $u \in \finwords$, if $\N(u) = \$ $, then $\D(u) = \$$, where $\$ \in \setnocond{+, -}$.
Clearly the size of $\D$ is bounded by the size of the TS, i.e. $0 < n \leq \size{\states}$, since we can obtain a DFA from the dDFA by simply using $\D^+$ (or the complement of $\D^-$). 
Nevertheless, we aim at finding the minimal such integer $n$.

To do this, we encode our problem as a SAT problem such that there is a separating complete DFA $\D$ with $n$ states if, and only if, the SAT problem is satisfiable.
We apply the standard propositional encoding~\citep{DBLP:conf/atva/Neider12,DBLP:conf/nfm/NeiderJ13,DBLP:conf/lata/UlyantsevZS15,DBLP:conf/lata/ZakirzyanovMIUM19}.
For simplicity, we let $\setnocond{0, \cdots, n-1}$ be the set of states of $\D$, such that $0$ is the initial one.
To encode the target DFA $\D$, we use the following variables:
\begin{itemize}
    \item the transition variable $e_{i, a, j}$ denotes that $i \pathto{a}{\trans} j$ holds, i.e.\ $e_{i, a, j}$ is true if, and only if, there is a transition from state $i$ to state $j$ over $a \in \alphabet$, and
    \item the acceptance variable $f_{i}$ denotes that $i\in F$, i.e.\ $f_i$ is true if, and only if, the state $i$ is an accepting one.
\end{itemize}

Once the problem is satisfiable, from the values of the above variables, it is easy to construct the DFA $\D$.
To that end, we need to tell the SAT solver how the DFA should look like by giving the constraints encoded as clauses.
For instance, to make sure the result DFA is indeed deterministic and complete, we need following constraints:
\begin{enumerate}
    \item[D1] Determinism:
    For every state $i $ and a letter $a \in \alphabet$ in $\D$, we have that $\neg e_{i, a, j} \vee \neg e_{i, a, k}$ for all $ 0 \leq j < k < n$.

    \item[D2] Completeness:
    For every state $i$ and a letter $a \in \alphabet$ in $\D$, $\bigvee_{0\leq j < n} e_{i, a, j}$ holds.
\end{enumerate}

Moreover, to make sure the obtained DFA $\D$ is separating for $\N$, we also need to perform the product of the target DFA $\D$ and $\N$.
In order to encode the product, we use extra variables $d_{p, i}$, which indicates that the state $p$ of $\N$ and the state $i$ of $\D$ can both be reached on some word $u$.
The constraints we need to enforce that $\D$ is separating for $\N$ are formalised as below:
\begin{enumerate}
    \item[D3] Initial condition:
    $d_{\init, 0}$ is true for all $\init \in \inits$. ($0$ is the initial state of~$\D$.
    \item[D4] Acceptance condition: for each state $i$ of $\D$,
    \begin{itemize}
        \item[D4.1] Accepting states:  $d_{p, i} \Rightarrow f_i$ holds for all $p \in A$
        \item[D4.2] Rejecting states: $d_{p, i} \Rightarrow \neg f_i$ holds for all $p \in R$;
    \end{itemize}
    \item[D5] Transition relation: for a pair of states $i, j$ in $\D$,
    $d_{p, i} \land e_{i, a, j} \Rightarrow d_{p', j}$ where $p' = \trans(p, a)$ for all $p \in \states$ and $ a \in \alphabet$.
\end{enumerate}
Let $\phi^{\N}_n$ be the conjunction of all these constraints.
Then, we obtain the following theorem.
\begin{theorem}
    Let $\N$ be a dDFA of $S$ and $n \in \naturals$.
    Then $\phi^{\N}_n$ is satisfiable if, and only if, there exists a complete DFA $\D$ with $n$ states that is separating for $\N$. 
\end{theorem}

The formula $\phi^{\N}_n$ contains $\O(n^3\cdot \size{\alphabet} + n^2 \cdot \size{\states} \cdot \size{\alphabet}) $ constraints.

When looking for separating DFAs, the SAT solver may need to inspect multiple isomorphic DFAs that only differ in their state names for satisfiability.
If those isomorphic DFAs are not separating for $\N$, then the SAT solver still has to prove this for each DFA.
To reduce the search space, it suffices to check only a representative DFA for all isomorphic DFAs~\citep{DBLP:conf/lata/UlyantsevZS15}.
We will describe the representative DFA in the following section.

\subsection{SAT encoding of the representative DFA}

The representative DFA $\D$ is induced by restricting the structure of its breath-first search (BFS) tree $\tau$.
In our setting, an edge of the BFS tree is a directed connection from one node to another and it is labelled by a letter in $\alphabet$.
In this section, we need to enforce an order on the letters in $\alphabet$.
For simplicity, we let $\alphabet = \setnocond{0, \cdots, \textbf{c}-1}$ where $\textbf{c} > 0$.
Recall that the set of nodes in the tree is the set of states of $\D$.
Below we list the requirements of the BFS tree.
\begin{itemize}
    \item[A1] Minimal parent:
    \begin{itemize}
         \item[A1.1] If there is an edge from node $i$ to node $j$ in the BFS tree, then $i < j$ and $i$ is the minimal state that reaches $j$ via a transition in $\D$, and
         \item[A1.2] If $j$ is a child node of $i$ and $j+1$ is a child node of $k$, then, $i < k$. This is because we enforce that smaller children must have smaller parents.
    \end{itemize}
    \item[A2] Minimal letter edge: 
    \begin{itemize}
        \item[A2.1] If the edge from $i$ to $j$ is labelled with letter $a$ in the BFS tree, $a$ must be the minimal letter from $i$ to $j$ in $\D$, and
        \item[A2.2] If there are edges from $i$ to $j$ over $a_1$ and to $k$ over $a_2$ in the BFS tree, and $a_1 < a_2$, then $j < k$. 
    Note that it is impossible for $a_1$ and $a_2$ to be equal since $\D$ is deterministic.
    \end{itemize}
    
\end{itemize}

In order to encode the above requirements, we need the following three types of boolean variables.
For a pair of states $0 \leq i, j < n$ and a letter $0 \leq a < \textbf{c}$: 
\begin{itemize}
    \item The edge variable $m_{i, a, j}$ of the BFS tree denotes that $m_{i, a, j}$ is true if, and only if, the BFS tree has an edge from node $i$ to node $j$ labelled with the letter $a$.
    \item The parent variable $p_{j, i}$ indicates that $p_{j, i}$ is true if, and only if, node $j$ is a child node of node $i$ in the BFS tree.
    \item The transition variable $t_{i,j}$ indicates that $t_{i,j}$ is true if, and only if, there is a transition from state $i$ to state $j$ in $\D$.
\end{itemize}

Now we can give the following constraints that are needed to represent the requirements of the BFS tree below.
\begin{enumerate}
    \item[B1] Minimal parent:
     \begin{itemize}
         \item[B1.1] for two nodes $0 \leq i < j < n$ in the BFS tree, we have $p_{j,i} \Leftrightarrow t_{i,j} \land \bigwedge_{0 \leq k < i} \neg t_{k, j}$, and
         \item[B1.2] for every triple $0 \leq k < i < j < n$, we have $p_{j, i} \Rightarrow \neg p_{j+1, k}$.
     \end{itemize}

    \item[B2]
    Minimal letter edge: 
    \begin{itemize}
        \item[B2.1] for every pair $0 \leq i < j < n$ in the BFS tree and a letter $a \in \alphabet$, we have     $m_{i, a, j} \Leftrightarrow (e_{i, a, j} \land \bigwedge_{0 \leq b < a} \neg e_{i, b, j} $).
        \item[B2.2] for every pair $0 \leq i < j < n$ of nodes and a pair $0 \leq a < b < \textbf{c}$ of letters, we have $p_{j, i} \land p_{j+1, i} \land m_{i, b, j} \Rightarrow \neg m_{i, a, j+1}$.
    \end{itemize}
    \item[B3] Edge consistency: for each pair $0 \leq i < j < n$, $t_{i, j} \Leftrightarrow \bigvee_{0 \leq a < \textbf{c}} e_{i, a, j}$ holds.
    \item[B4] Existence of a parent: for each node $0 \leq i < n$ in the BFS tree, we have that $\bigvee_{0 \leq j < i} p_{j, i}$ holds.
\end{enumerate}
Let $\phi^{\tau}_n$ be the conjunction of all these constraints.
Then we obtain the following theorem.
\begin{theorem}
    Let $\N$ a dDFA of $S$ and $n \in \naturals$.
    Then, $\phi^{\N}_n \land \phi^{\tau}_n$ is satisfiable if, and only if, there exists a separating DFA $\D$ with $n$ states, with respect to $\N$. 
\end{theorem}

We remark that the formula $\phi^{\tau}_n$ contains $\O(n^3 + n^2\cdot \size{\alphabet}^2) $ constraints.

\section{More experimental results}\label{apd:more-results}

In this section we provide additional experiments.
Table~\ref{tab:full-results} extends Table~\ref{tab:exp-short} by reporting also the minimum and maximum running time for each technique. Since the samples are randomly generated, the minimum and maximum values can show irregular peaks, while the average tends to converge to the timeout value.

In Table~\ref{tab:full-automata-results} we report, on the same benchmark as the previous table, the minimum, maximal, and average size of the APTA for \dfaind~and \dfaid~(they build the exact same APTA) and the size of the 3DFA for the two variants of \tool, namely the double DFA and 3DFA.

\begin{table}[]
    \centering
     \scalebox{0.7}{
    \begin{tabular}{ c | c c c c | c c c c | c c c c }\toprule
    \multicolumn{1}{c}{} & \multicolumn{4}{c}{DFA-Inductor} & \multicolumn{4}{c}{dDFA-MIN} & \multicolumn{4}{c}{DFA-MIN} \\
    \hline
    N & min & avg & max & solved\% & min & avg & max & solved\% & min & avg & max & solved\% \\
    \hline
    4 & 0.01 & 0.12 & 0.20 & 100 & 0.01 & 0.03 & 0.49 & 100 & 0.01 & 0.02 & 0.09 & 100 \\
    5 & 0.01 & 0.29 & 0.62 & 100 & 0.01 & 0.06 & 1.01 & 100 & 0.01 & 0.05 & 0.22 & 100 \\
    6 & 0.02 & 0.67 & 1.43 & 100 & 0.01 & 0.30 & 1.82 & 100 & 0.01 & 0.18 & 1.91 & 100 \\
    7 & 0.02 & 1.81 & 5.86 & 100 & 0.03 & 0.80 & 7.09 & 100 & 0.04 & 0.73 & 3.61 & 100 \\
    8 & 0.02 & 3.57 & 10.26 & 100 & 0.01 & 1.29 & 9.06 & 100 & 0.01 & 1.25 & 4.53 & 100 \\
    9 & 0.02 & 10.84 & 39.84 & 100 & 0.02 & 3.83 & 11.94 & 100 & 0.01 & 3.78 & 15.71 & 100 \\
    10 & 3.85 & 50.91 & 199.04 & 100 & 0.69 & 17.88 & 62.40 & 100 & 0.63 & 16.80 & 57.27 & 100 \\
    11 & 1.90 & 154.73 & 691.79 & 100 & 0.37 & 55.12 & 263.82 & 100 & 0.33 & 59.46 & 288.10 & 100 \\
    12 & 15.86 & 399.52 & 1200 & 96 & 4.30 & 144.27 & 536.74 & 100 & 5.12 & 162.39 & 705.19 & 100 \\
    13 & 84.84 & 850.04 & 1200 & 74 & 28.61 & 390.10 & 1200 & 99 & 15.86 & 418.62 & 1200 & 97 \\
    14 & 214.76 & 1125.59 & 1200 & 19 & 77.56 & 809.88 & 1200 & 76 & 68.03 & 861.10 & 1200 & 69 \\
    15 & 587.66 & 1182.98 & 1200 & 6 & 225.56 & 1060.18 & 1200 & 37 & 218.22 & 1062.02 & 1200 & 34 \\
    16 & 17.22 & 1188.17 & 1200 & 1 & 9.55 & 1167.58 & 1200 & 4 & 5.06 & 1164.02 & 1200 & 5 \\
    \bottomrule
\end{tabular}
}
\caption{Full experimental results on generated benchmarks}\label{tab:full-results}
\end{table}

\begin{table}[]
    \centering
    \begin{tabular}{ c | c c c | c c c | c c c }\toprule
    \multicolumn{1}{c}{} & \multicolumn{3}{c}{DFA-Inductor/DFA-Identify} & \multicolumn{3}{c}{dDFA-MIN} & \multicolumn{3}{c}{DFA-MIN} \\ \hline
    N & min & avg & max & min & avg & max & min & avg & max \\ \hline
    4 & 526 & 549 & 570 & 189 & 230 & 264 & 188 & 207 & 227\\
    5 & 650 & 675 & 695 & 221 & 281 & 321 & 220 & 248 & 270\\
    6 & 798 & 822 & 860 & 267 & 338 & 374 & 266 & 299 & 322\\
    7 & 935 & 963 & 994 & 312 & 393 & 437 & 311 & 347 & 373\\
    8 & 1073 & 1097 & 1126 & 347 & 439 & 481 & 346 & 389 & 418\\
    9 & 1195 & 1225 & 1260 & 383 & 485 & 535 & 382 & 425 & 449\\
    10 & 1322 & 1349 & 1378 & 454 & 528 & 571 & 427 & 461 & 485\\
    11 & 1459 & 1496 & 1542 & 503 & 580 & 622 & 474 & 510 & 530\\
    12 & 1602 & 1643 & 1683 & 563 & 638 & 679 & 520 & 559 & 586\\
    13 & 1748 & 1791 & 1836 & 618 & 690 & 725 & 556 & 604 & 632\\
    14 & 1887 & 1930 & 1968 & 664 & 741 & 777 & 605 & 648 & 670\\
    15 & 2018 & 2063 & 2098 & 716 & 786 & 821 & 651 & 690 & 719\\
    16 & 2159 & 2196 & 2230 & 647 & 827 & 873 & 632 & 726 & 755\\
    \bottomrule
\end{tabular}
\caption{Full automata size results on generated benchmarks before the SAT minimisation}\label{tab:full-automata-results}
\end{table}

In the following figures we propose a pairwise comparison of all the techniques on the minimisation time. Figure~\ref{fig:dfavsind} shows \tool~and \dfaind.
Interestingly, in Figure~\ref{fig:dfavsddfa}, the two variants of \tool~do not overlap, there are cases up to two times easier to minimise for the double DFA variant and cases up to three times easier for the 3DFA variant, which, as shown in Table~\ref{tab:full-results}, on average is slightly faster.
In Figure~\ref{fig:ind2vsddfa} we compare the double DFA variant of \tool~and \dfaind~2 which employs a different, and more efficient, encoding of DFA called tightDFS. This heuristic can also be applied to our tool. Although the application of tight encoding, \dfaind~2 is still slower that \tool~up to four times slower. However, in some cases it can be slightly faster or solve cases on which dDFA hit the timeout (seven cases with N=14). Similarly, Figure~\ref{fig:ind2vsddfa} compares the 3DFA variant of \tool~and \dfaind~2. In this case the number of DFAs not solved by \tool~are eight, all with N=14.
Finally, Figures~\ref{fig:idvddfa} and~\ref{fig:idvsddfa} show the comparison of \dfaid~and the two variants of \tool.

\begin{figure}
    \centering
    \scalebox{0.8}{
    \includegraphics{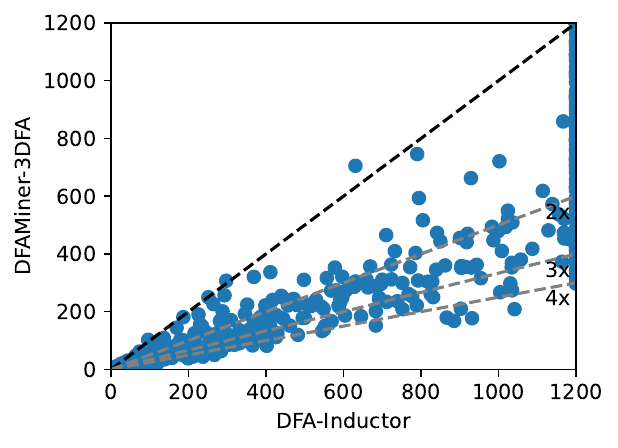}
    }
    \caption{Minimisation time of \tool~with double DFA and DFA inductor on 1,300 DFAs.}
    \label{fig:dfavsind}
\end{figure}

\begin{figure}
    \centering
    \scalebox{0.8}{
    \includegraphics{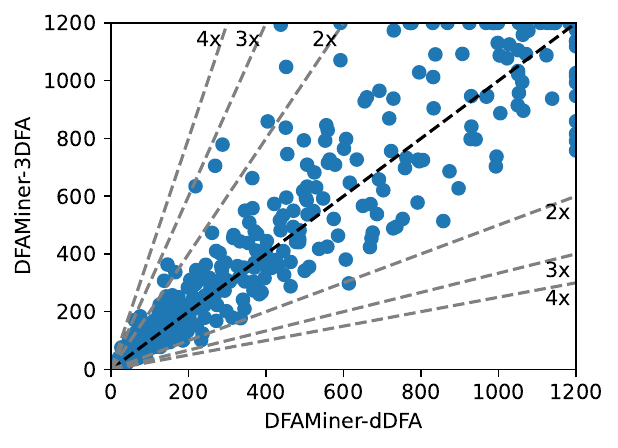}
    }
    \caption{Minimisation time of the two \tool~variants: the double DFA and 3DFA.}
    \label{fig:dfavsddfa}
\end{figure}

\begin{figure}
    \centering
    \scalebox{0.8}{
    \includegraphics{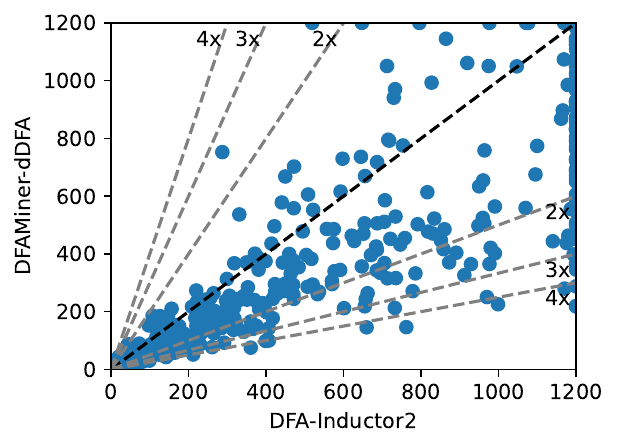}
    }
    \caption{Minimisation time of the \dfaind~2 and \tool~with dDFA.}
    \label{fig:ind2vsddfa}
\end{figure}

\begin{figure}
    \centering
    \scalebox{0.8}{
    \includegraphics{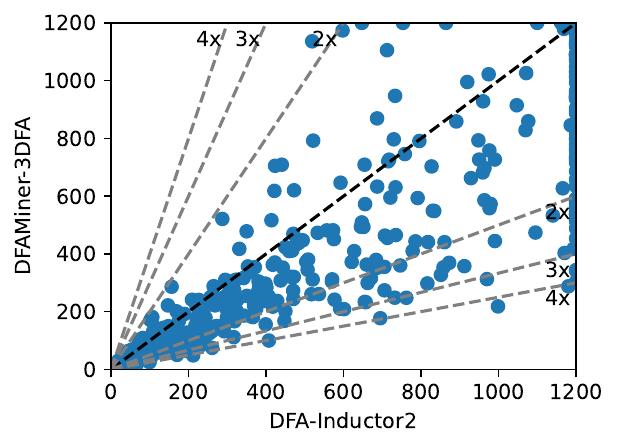}
    }
    \caption{Minimisation time of the \dfaind~2 and \tool~with 3DFA.}
    \label{fig:dfavsddfa}
\end{figure}

\begin{figure}
    \centering
    \scalebox{0.8}{
    \includegraphics{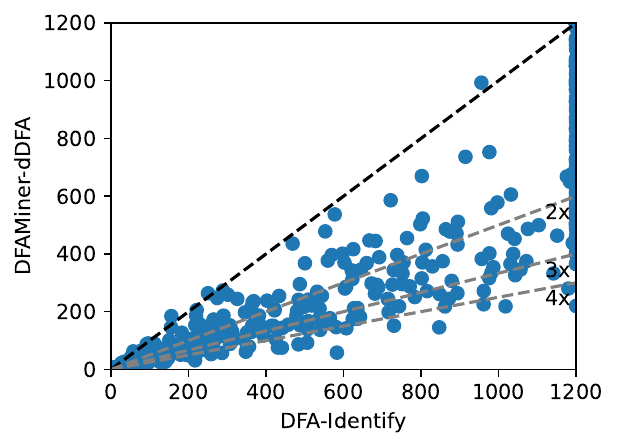}
    }
    \caption{Minimisation time of the \dfaid~and \tool~with dDFA.}
    \label{fig:idvddfa}
\end{figure}

\begin{figure}
    \centering
    \scalebox{0.8}{
    \includegraphics{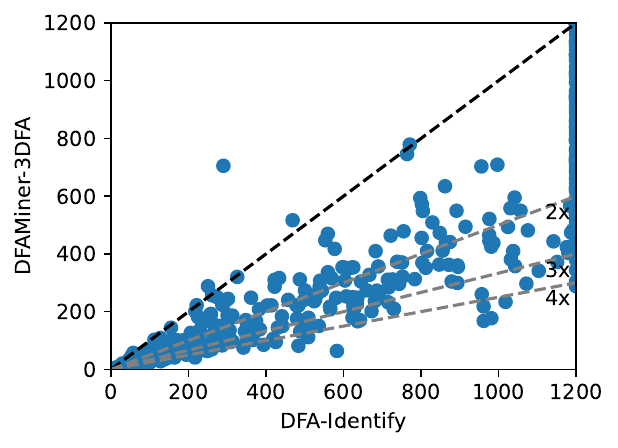}
    }
    \caption{Minimisation time of the \dfaid~and \tool~with 3DFA.}
    \label{fig:idvsddfa}
\end{figure}

\subsection{Size growth of APTA, dDFA and 3DFA in parity game solving}
\label{app:apta-3dfa-size}
\begin{table}[]
    \centering
    \begin{tabular}{ c c | c | c | c | c | c c }
    Colours & Length & dDFA & 3DFA & APTA & Samples & \#Pos & \#Neg\\
    \hline
    2 & 3 & 12 & 8 & 15 & $2^3$ & 3 & 5\\ \hline 

    3 & 4 & 28 & 23 & 111& $3^4$ & 51 & 20\\  
    3 & 5 & 38 & 33 & 266& $3^5$ &  130 & 31\\ \hline

    4 & 5 & 84 & 82 & 1,083& $4^5$ & 274 & 488\\  
    4 & 6 & 117 & 122 & 3,311& $4^6$ & 669 & 1,599\\
    4 & 7 & 150 & 155 & 10,076& $4^7$ & 1,645 & 5,235\\ \hline
 
    5 & 6 &  269 & 301 & 13,634& $5^6$ & 7,233 & 3,067\\  
    5 & 7 & 372 & 438 & 53,277&$5^7$ & 30,332 & 9,625\\
    5 & 8 & 475 & 541  & 209,721& $5^8$ & 127,194 & 30,456\\
    5 & 9 & 578 & 644 & 835,954& $5^9$ & 533,305 & 97,228\\
    5 & 10 & 681 & 747 & 3,369,694& $5^{10}$ & 2,236,023 & 312,568\\
    5 & 11 & 784 & 850 & 13,704,486& $5^{11}$ & 9,375,269 & 1,009,941\\ \hline
    6 & 7 & 986 & 1,279 & 199,397& $6^7$ & 53,556 & 104,095\\ 
    6 & 8 & 1,349& 1,807 & 930,870& $6^8$ & 216,298 & 517,590\\ 
    6 & 9 & 1,712& 2,170 & 4,369,362 & $6^9$ & 878,823 & 2,573,621\\
    6 & 10 & 2,075 & 2,533 & 20,689,546& $6^{10}$ & 3,595,591 & 12,795,642\\
    6 & 11 & 2,438 & 2,896 & - & $6^{11}$ & 14,799,059 & 63,616,339\\
    6 & 12 & 2,801& 3,259 & - & $6^{12}$ & 61,192,124 & 316,286,133\\
    6 & 13 & 3,164 & 3,622 & - & $6^{13}$ & 253,881,602 & 1,572,522,807\\
    6 & 14 & 3,527 & 3,985 & - & $6^{14}$ & 1,055,948,048 & 7,818,368,374\\
    6 & 15 & 3,890& 4,348 & - & $6^{15}$ & 4,399,883,736 & 38,871,920,470\\
    6 & 16 & 4,253& 4,711 & - & $6^{16}$ & 18,357,865,115 & 193,266,275,998\\

    \hline
    
\end{tabular}
    \caption{The size comparison between APTAs, 3DFAs and dDFAs for parity game solving. Some of the missing data for APTA is due to the fact that there is no enough RAM memory to build it.}
    \label{tab:parity-apta-3dfa}
\end{table}

In Table~\ref{tab:parity-apta-3dfa}, we can see that the number of positive and negatives samples grow exponentially, even if they eventually take up below 20\% of all samples.
On the other hand, we can see that for a fixed colour number, the sizes of dDFAs and 3DFAs grow constantly when increasing the word length by $1$.
For instance, for the colour number $6$, the size of 3DFA increases by at most 528 and eventually by 363, and dDFA by 363 when increasing the length by 1.
A huge save in the number of states representing the samples thus leads to a significantly better performance in solving the Min-DFA inference problems for parity games.

\subsection{Encoding of safety automata}
\label{app:parity}

In this section we outline the encoding employed for minimising safety automata from parity samples. In particular, we introduce additional contraints to reduce the space of the search of the DFA.
Because of the huge number of samples, we have tweaked the encodings described in the Section \ref{apd:sat-encoding} for efficiency, taking properties of separating automata for parity games into account.

Let \textbf{c} be the number of colours in the parity game, i.e. the set of the colours is $[0,\textbf{c}-1]$. We use the previously defined transition variables from a state $i$ to $j$ reading a letter (colour) $a$ as $e_{i, a, j}$ and the acceptance variables $f_i$, with $ 0 \leq i, j < n$ and $ 0 \leq a < \textbf{c}$.
Moreover, to make the SAT problem easier, we set $0$ as the initial state and $n-1$ as the sink state (for safety or co-safety acceptance).
Let $\ell$ be a positive integer.
We denote by $\odd{\ell} $ (respectively, $\even{\ell}$) the set of odd (respectively, even) numbers in the set $[0,\ell - 1]$.
By $\opp(\textbf{c})$ we denote the set of colours having the opponent priority as the highest colour, i.e. if $\textbf{c}-1$ is even, then $\opp(\textbf{c}) = \odd{\textbf{c}}$, and $\opp(\textbf{c}) = \even{\textbf{c}}$ otherwise.
The additional constraints are as follow:
\begin{enumerate}

    \item[P1] Initial state loops. Whenever the initial state $0$ reads a colour of the same parity as the highest one, then it loops over itself.
    Therefore, if $\textbf{c}-1$ is even, then we have $\bigwedge_{a \in \even{\textbf{c}}} e_{0, a, 0}$, otherwise $\bigwedge_{a \in \odd{s}} e_{0, a, 0}$ holds.
    
    \item[P2] Initial state outgoing transitions. Whenever the initial state $0$ reads a colour of the opponent parity as $\textbf{c}-1$, it reaches a different state than $0$ and $n-1$.
    Hence, we have $\bigwedge_{a \in \opp(\textbf{c})}(\bigvee_{0 < i < n-1} e_{0, a, i})$.

    \item[P3] Sink state incoming transitions\footnote{This constraint may conflict with the requirement R4 for the representative DFA. So, we need to drop those constraints of the representative DFA for state $n-1$ that require the incoming transition to state $n-1$ to be over the maximal letter.}. Whenever a state $i\neq n-1$ reads a colour of the same parity as the highest one, then it cannot reach the sink state $n-1$. The constraint is formalised as $\bigwedge_{0 \leq i < n-1} (\bigwedge_{a \not\in\opp(\textbf{c})}\neg e_{i, a, n-1})$.
    
    \item[P4] Reset transitions. Whenever a state $i\neq n-1$ reads the highest colour, then it reaches the initial state $0$.
    Then, we have that $\bigwedge_{0 \leq i < n-1} e_{i, \textbf{c}-1, 0}$ holds.

    \item[P5] Sink state loops. The sink state $n-1$ makes exception to the previous rule, it can only loop on itself.
    Formally, $\bigwedge_{0 \leq a < \textbf{c}} e_{n-1, a, n-1}$.

    \item[P6] No loops on opponent colours. No states but the sink can loop whenever reading a colour of the opponent parity as the highest one.
    Then, we have $ \bigwedge_{0 \leq i < n-1}(\bigwedge_{a \in \opp(\textbf{c})} \neg e_{i, a, i}$).
    
    \item[P7] Acceptance. If the highest colour is even, we look for a safety DFA (the sink state is the only rejecting one), otherwise we build a co-safety DFA (the only accepting state is the sink).
    Formally, if $\textbf{c}-1$ is even, then $(\bigwedge_{0 \leq i < n-1} f_{i} )\land \neg f_{n-1}$, otherwise $(\bigwedge_{0 \leq i < n-1} \neg f_{i}) \land f_{n-1}$.
    
\end{enumerate}

For best performance of parity game solving problems, it is better to turn on the safety encoding option when running \tool.

\end{document}